\documentclass[review,1p,times,10pt]{elsarticle}

\usepackage{amsmath,amssymb,amsthm,amsfonts,blkarray,pgfplots,tkz-fct}
\usepackage{graphicx,color,epsfig,multimedia,epstopdf,caption,chngcntr}
\usepackage{verbatim,url,multirow}
\pgfplotsset{compat=1.11}
\usetikzlibrary{intersections,backgrounds,trees}
\biboptions{sort&compress}
\usepackage[T1]{fontenc}
\usepackage[utf8]{inputenc}
\usepackage{float}

\def\NN{\mathbb N} 
\def\ZZ{\mathbb Z} 
\def\RR{\mathbb R} 

\everymath{\displaystyle}
\def\l{\left}
\def\r{\right}
\newcommand{\bb}[1]{\begin{equation}\label{#1}}
\newcommand{\ee}{\end{equation}}
\newcommand{\bbb}{\begin{eqnarray}}
\newcommand{\eee}{\end{eqnarray}}
\newcommand{\bbbb}{\begin{eqnarray*}}
\newcommand{\eeee}{\end{eqnarray*}}
\newcommand{\nnn}{\nonumber}
\newcommand{\no}{\noindent}
\def\R#1{$(\ref{#1})$}

\newtheorem{theorem}{Theorem}

\theoremstyle{remark}

\theoremstyle{definition}

\makeatletter
\def\thmhead@plain#1#2#3{%
  \thmname{#1}\thmnumber{\@ifnotempty{#1}{ }\@upn{#2}}%
  \thmnote{ {\the\thm@notefont#3}}}
\let\thmhead\thmhead@plain
\makeatother

\journal{}

\begin{document}

\begin{frontmatter}



\title{Anomalous Diffusion in One-Dimensional Disordered Systems: A Discrete Fractional Laplacian Method}


\author[ttu,ttu1]{J. L. Padgett}
\ead{joshua.padgett@ttu.edu}
\fntext[ttu]{Corresponding Author}
\address[ttu1]{Department of Mathematics and Statistics, Texas Tech University, Broadway and Boston, Lubbock, TX, 79409, USA}

\author[bu]{E. G. Kostadinova}

\author[bu,d]{C. D. Liaw}

\author[rice]{K. Busse}

\author[bu]{L. S. Matthews}

\author[bu]{T. W. Hyde}

\address[bu]{Center for Astrophysics, Space Physics, and Engineering Research and Department of Physics, One Bear Place, 97316, Baylor University, Waco, TX, 76798, USA}

\address[rice]{Department of Mathematics, Rice University, 6100 Main St., Houston, TX 77005, USA}

\address[d]{Department of Mathematical Sciences, University of Delaware, 311 Ewing Hall, Newark, DE, 19716, USA}

\begin{abstract}
This work extends the applications of Anderson-type Hamiltonians to include transport characterized by anomalous diffusion. Herein, we investigate the transport properties of a one-dimensional disordered system that employs the discrete fractional Laplacian, $(-\Delta)^s,\ s\in(0,2),$ in combination with results from spectral and measure theory. It is a classical mathematical result that the standard Anderson model exhibits localization of energy states for all nonzero disorder in one-dimensional systems. Numerical simulations utilizing our proposed model demonstrate that this localization effect is enhanced for sub-diffusive realizations of the operator, $s\in (1,2),$ while the super-diffusive realizations of the operator, $s\in (0,1),$ can exhibit energy states with less localized features. These results suggest that the proposed method can be used to examine anomalous diffusion in physical systems where strong correlations, structural defects, and nonlocal effects are present. 
\end{abstract}

\begin{keyword} 
Anderson localization \sep anomalous diffusion \sep discrete fractional Laplacian \sep spectral approach \sep disordered systems
\end{keyword}

\end{frontmatter}



\section{Introduction} 

The concept of localization was first studied by P. W. Anderson in 1958, when he suggested that sufficiently large impurities in a disordered medium could lead to the localization of electrons. This phenomenon is known as {\em Anderson localization} and has motivated various mathematical and physical studies in the past 60 years. The classical localization problem in one-dimension is well understood in systems with nearest-neighbor interactions, yet there are a multitude of open questions related to transport driven by correlated and nonlocal effects. Herein, we propose a non-local model for studying one-dimensional anomalous transport which provides new perspectives on numerous physical and mathematical problems. This new model allows for the consideration of systems which exhibit exotic transport properties and interactions. In particular, we are interested in a generalized notion of diffusion known as {\em anomalous diffusion}.

Diffusion is a persistent random walk characteristic of diverse systems, such as neutrons in nuclear reactors \cite{osti_4074688}, stock market prices \cite{bachelier2011louis}, and pollen particles suspended in fluids \cite{nelson1967dynamical}. In the standard diffusion regime, the mean squared displacement of an ensemble of moving particles increases linearly in time, {\em i.e.} $\langle (x-x_0)^2\rangle \sim t^\beta,$ where $\beta=1.$ However, nonlinear mean squared displacement, characterized by exponents $\beta\neq 1,$ is also possible, yielding two distinct examples of {\em anomalous transport}: {\em subdiffusion}, when $\beta\in(0,1),$ and {\em superdiffusion}, when $\beta>1.$ Anomalous diffusion has been analyzed theoretically and observed experimentally in various physical systems, including amorphous semiconductors, glasses, porous media, granular matter, and plasmas \cite{bouchaud1990anomalous,10.1007/BFb0106838,binder1999anomalous,UPADHYAYA2001549,palombo2013structural,drummond1962anomalous,tarasov2006fractional}. Such diffusive behavior has also been related to complex processes such as turbulence, biological cell motility, and superconductivity \cite{shlesinger1993strange,0305-4470-37-31-R01,hansen1999dispersion}. In the presence of disorder, many of the aforementioned systems are known to exhibit localization properties, although the definitions of localization in each study strongly depend on the underlying physical system. The current study aims to investigate the effect of anomalous diffusion on the transport properties of one-dimensional disordered media. This research is the first of its kind in that it combines the spectral approach developed in \cite{Liaw2013} and the known discrete fractional Laplacian results from \cite{ciaurri2015fractional}. The proposed method employs the spectral definition of localization, which will be further discussed in Section 3.

The time evolution of a physical system is governed by the repeated application of a Hamiltonian operator to a given initial state. A typical Hamiltonian for modeling transport via classical diffusion in a one-dimensional disordered system is given by the discrete random Schrödinger operator of the form
\bb{std1}
H_\epsilon \mathrel{\mathop:}= -\Delta + \sum_{i\in\ZZ} \epsilon_i \,\langle \cdot,\delta_i\rangle\, \delta_i,
\ee
where $\Delta$ is the discrete Laplacian in one-dimension, with $\langle \cdot,\cdot\rangle$ being the $\ell^2(\ZZ)$ inner product, $\delta_i$ are the standard Kronecker delta functions defined on $\ZZ,$ and $\epsilon_i$ are random variables taken to be independently and identically distributed (i.i.d.)~according to the uniform distribution on $[-c/2,c/2],$ with $c>0.$ The operator given in \R{std1} does not allow one to model systems exhibiting anomalous diffusion, hence, we extend this method by considering the {\em random discrete fractional Schr{\"o}dinger operator}, formally given by
\bb{schro1}
H_{s,\epsilon} \mathrel{\mathop:}= (-\Delta)^s + \sum_{i\in\ZZ} \epsilon_i \,\langle \cdot, \delta_i\rangle\, \delta_i,
\ee
for some $s\in (0,2).$

The operator $(-\Delta)^s$ is the {\em discrete fractional Laplacian} and will be defined in the following section. The operator $(-\Delta)^s$ can be used to describe the {\em non-local} motion of an electron in a one-dimensional chain with atoms located at all integer lattice points in $\ZZ.$ When $s=1,$ the operator in \R{schro1} reduces to the classical random discrete Schr{\"o}dinger operator, given in \R{std1}, studied in \cite{jakvsic2000spectral}. However, when $s\neq 1,$ this operator considers the possibility of electrons jumping to non-neighboring lattice points, which corresponds to an {\em anomalous-type diffusion} process. The perturbation $\textstyle\sum_{i\in\ZZ} \epsilon_i \,\langle \cdot, \delta_i\rangle\, \delta_i$ can be used to model random displacements of the atoms located at the lattice points. This perturbation is almost surely a non-compact operator, which means that classical perturbation theory cannot be applied (for more details see \cite{birman2012spectral,kato2013perturbation}). Here, the parameter $c$ is interpreted as the strength of the disorder at the lattice points. Finally, it is worth emphasizing that the nonlinear mean squared displacement exponent $\beta$, while related, differs from the fractional power of the discrete Laplace operator. For the particular operator considered in this study, $\beta$ is asymptotically proportional to $s^{-1}.$ More details regarding this relationship will be provided in Section 2.3.

In the classical setting, $s=1,$ it has been shown via the spectral method (see Section 3 for details) that localization is expected for all disorders $c>0$ \cite{Liaw2013}. The spectral method was introduced in \cite{Liaw2013}, where it was also used to numerically confirm the existence of extended states for the two-dimensional discrete random Schr\"odinger operator for weak disorder. Applications to other underlying geometries, such as the square, hexagonal, triangular lattice in two-dimensions, and the three-dimensional square lattice, were explored in \cite{1751-8121-47-30-305202,2053-1591-3-12-125904,PhysRevB.96.235408,kostadinova2018transport,PhysRevB.99.024115}. The unperturbed operator used in these papers was the classical Laplacian ($s=1$). The spectral approach allows for the development of efficient computational techniques that can provide direction and intuition for analytic results in both mathematics and physics. Moreover, the physical interpretation of this approach has been recently established \cite{2053-1591-3-12-125904}.
Other theoretical expectations and results for the continuous counterpart of \R{schro1}, in the case where $s\in(0,1),$ have been established in \cite{laskin2002fractional}.

Herein, we extend the application of the spectral approach to the study of transport guided by nonlocal interactions. Recent numerical studies have suggested that delocalized transport behavior can potentially occur in media with correlated disorders \cite{lazo2010conducting,shima2005breakdown,shima2006metal}. In contrast, the present work employs the operator given by \R{schro1} to model the nonlocal interaction, while assuming random uncorrelated on-site energies. Using this model, we numerically investigate localization properties of the operator given by \R{schro1}. By considering various fractional powers of the Laplacian, we demonstrate enhanced localized behavior for $s\in(1,2)$ and enhanced transport for $s\in(0,1).$ These observations have interesting implications both mathematically and physically, thus yielding exciting new avenues of research. The current study will undoubtedly motivate the consideration of localization in the presence of anomalous diffusion with respect to the more restrictive definitions, such as {\em dynamic localization}.

This article is organized as follows. Section 2 provides relevant theoretical background regarding the discrete fractional Laplacian and its associated nonlocal weights. Section 3 details the spectral approach employed for studying the transport behavior of the newly proposed model. Section 4 outlines in detail the numerical method used in our simulations. We provide computational results that validate the proposed method and elucidate numerous analytical properties of the random discrete fractional Schr{\"o}dinger operator. Finally, concluding remarks and a discussion of possible future projects are provided in Section 5. 



\section{Theoretical Background}

In the following sub-sections we present the necessary theoretical background for the discrete fractional Laplacian. Section 2.1 introduces the discrete fractional Laplacian and outlines some necessary results for the construction of the numerical method. Novel results regarding higher fractional powers of the Laplacian are provided. In Section 2.2, we briefly provide a description of physical interpretations of the discrete fractional Laplacian relevant to the current work and Section 2.3 provides a brief motivation of the relationship between the fractional power of the Laplace operator and its associated nonlinear mean squared displacement.

\subsection{Discrete Fractional Laplacian}

The fractional Laplacian has been studied in mathematics for nearly a century. Understood as the classical Laplacian raised to positive powers, this operator has received attention in potential theory, fractional calculus, harmonic analysis, and probability theory \cite{10.2307/1993412,bogdan1999potential,baleanu2012fractional,valdinoci2009long}. However, only recently has the operator garnered attention in the fields of differential equations and physics. Due to this newfound interest, the fractional Laplacian has become one of the most researched mathematical objects of the past decade. Classically, only fractional powers $s\in (0,1)$ have been considered, and in this case, one can define the fractional Laplacian on $\RR^d$ as the hyper-singular integral given by
\bb{def1}
(-\Delta)^s u(x)\mathrel{\mathop:}= c_{d,s}\,\lim_{\varepsilon\to 0^+}\int_{\RR^d\backslash B_\varepsilon(x)}\frac{u(x)-u(\xi)}{|x-\xi|^{d+2s}}\,d\xi,
\ee
where $x\in\RR^d,$ $B_\varepsilon(x)$ is the $d-$dimensional ball of radius $\varepsilon>0$ centered at $x\in\RR^d,$ and $c_{d,s}$ is some normalization constant. The operator can also be defined as a pseudo-differential operator via its Fourier transform, {\em i.e.},
\bb{def2}
\widehat{(-\Delta)^s}u(\xi) = |\xi|^{2s}\hat{u}(\xi).
\ee

The recent interest in the fractional Laplacian has been a direct consequence of the revolutionary work by Caffarelli and Silvestre, who demonstrated in \cite{doi:10.1080/03605300600987306} that one may study $(-\Delta)^s,\ s\in(0,1),$ via the Dirichlet-to-Neumann operator associated with a particular extension problem. The Dirichlet-to-Neumann operator is a particular example of the Poincar{\'e}-Steklov operator and maps the values of a harmonic function on the boundary of some domain to the normal derivative values
of the same function on the same boundary. Caffarelli and Silvestre's approach provided an extension of a well-known result regarding the square root of the Laplacian as it arose in fluid dynamics and finance (see \cite{cabre2010positive} and the references therein). That is, for $s\in (0,1),$ they showed that
\bb{dlap1}
(-\Delta)^su(x) = c_s \lim_{t\to 0^+} t^{1-2s}v_t(x,t),
\ee
where $v(x,t)\,:\,\RR^d\times \RR \to \RR_+$ is the solution to the following Bessel-type problem
\bb{bprob}
\l\{\begin{array}{rcll}
v_{tt}(x,t) + \frac{1-2s}{t}v_t(x,t) + \Delta v(x,t) & = & 0 , & x\in\RR^d,\ t > 0,\\
v(x,0) & = & u(x), & x\in \RR^d,
\end{array}\r.
\ee
and $c_s \mathrel{\mathop:}= 2^{2s-1}\Gamma(s)/\Gamma(1-s).$ Thus, one may study the highly nonlocal fractional Laplacian operator by considering the local problem \R{bprob}. While \R{bprob} is posed in one higher dimension and exhibits either a singular or degenerate nature depending on the value of $s,$ it is amenable to classical analytical and numerical techniques. Recent work by Chen, Lei, and Wei has demonstrated that similar extension problems may be derived for higher fractional powers of the Laplacian, though these extensions have not yet been employed in approximating solutions to higher-order problems \cite{Chen2018}.

In this study, we investigate the fractional powers of the discrete Laplacian for exponents $s\in(0,2).$ While numerous definitions of such an operator exist (for instance, see \cite{kwasnicki2017ten}), we begin by introducing existing results for $s\in(0,1)$ (see \cite{ciaurri2016nonlocal,ciaurri2015fractional}) and then expand them to the case where $s\in(0,2).$ Let $u\,:\,\ZZ\to\RR$ with $u_n\mathrel{\mathop:}= u(n),\ n\in\ZZ.$ We then define the discrete Laplacian on $\ZZ$ as
\bb{dlap}
\Delta u_n \mathrel{\mathop:}= u_{n+1}-2u_n+u_{n-1}.
\ee
For the discrete Laplacian, an extension problem similar to \R{bprob} can be generated  and solved uniquely, with its bounded solution given by
\bb{extsol}
v(x,t) = \frac{1}{\Gamma(s)}\int_0^\infty z^{s-1}e^{-t^2/4z}e^{-z\Delta}(-\Delta)^su(x)\,dz,
\ee
where $e^{-z\Delta}$ is the standard semigroup generated by the discrete Laplacian on $\ZZ$ \cite{ciaurri2016nonlocal,meichsner2017fractional,doi:10.1080/03605301003735680,Gale2013}. Explicit calculation then yields, via \R{dlap1}, the following representation
\bb{dlap2}
(-\Delta)^s u(x) = \frac{1}{\Gamma(-s)}\int_0^\infty z^{-s-1}\l(e^{-z\Delta}-I\r)u(x)\,dz,
\ee
where $I$ is the identity operator. The expression \R{dlap2} is often used as the definition of fractional powers of the Laplacian \cite{ciaurri2016nonlocal,ciaurri2015fractional}.
It is also important to note that the formulations outlined above result in operators which are {\em not} discretizations of the continuous fractional Laplace operator (for instance, see \cite{zoia2007fractional,huang2014numerical,duo2018novel}).
From \R{dlap2} and standard results regarding the discrete Laplacian semigroup (see \cite{Ciaurri2017} and references therein), the following theorem has been developed \cite{ciaurri2015fractional}. 

\begin{theorem}[\cite{ciaurri2015fractional}]
For $s\in (0,1),$ we define
$$\ell_{s}\mathrel{\mathop:}= \l\{u\,:\,\ZZ\to\RR\,:\,\|u\|_{\ell_{s}}\mathrel{\mathop:}= \sum_{n\in\ZZ}\frac{|u_n|}{(1+|n|)^{1+ 2s}}<\infty\r\}.$$
\begin{itemize}
\item[i.] For $u\in \ell_{s}$ we have 
\bb{solform}
(-\Delta)^su_n = \sum_{m\in\ZZ;\,m\neq n} \l(u_n - u_m\r)K_s(n-m),
\ee
where the discrete kernel is given by
\bb{kern}
K_s(m)\mathrel{\mathop:}= \l\{\begin{array}{ll}
\frac{4^s\Gamma(1/2+s)}{\sqrt{\pi}|\Gamma(-s)|}\cdot\frac{\Gamma(|m|-s)}{\Gamma(|m|+1+s)}, & m\in\ZZ\backslash\{0\},\\
0, & m=0.
\end{array}\r.
\ee

\item[ii.] For $s\in (0,1)$ there exists constants $0 < c_s \le C_s$ such that, for any $m\in\ZZ\backslash\{0\},$ 
\bb{ineq1}
\frac{c_s}{|m|^{1+2s}} \le K_s(m) \le \frac{C_s}{|m|^{1+2s}}.
\ee

\item[iii.] If $u\in \ell_0,$ then $\lim_{s\to 0^+}(-\Delta)^s u_n = u_n.$

\item[iv.] If $u$ is bounded, then $\lim_{s\to 1^-} (-\Delta)^s u_n = -\Delta u_n.$
\end{itemize}
\end{theorem}

From Theorem 1, we are able to discern many useful properties which are important for both the theoretical analysis and computational procedures employed in problems involving the fractional Laplacian. In the ensuing experiments, our calculations will be based on \R{solform} and its various approximations, the details of which are contained in Section 4. 

It now remains to discuss the case when $s\in (1,2).$ In \cite{ros2014pohozaev} it was shown that for $s>1,$ we have
\bb{s>1}
(-\Delta)^su_n = (-\Delta)^{s-1}(-\Delta)u_n.
\ee
Using \R{s>1}, we have the following theorem for $s\in (1,2).$ The proof for Theorem 2, parts {\em ii.}-{\em iv.}, are similar to those in \cite{ciaurri2016nonlocal} but we include them for completeness.

\begin{theorem}
For $s\in (1,2),$ we define
$$\ell_{s}\mathrel{\mathop:}= \l\{u\,:\,\ZZ\to\RR\,:\,\|u\|_{\ell_{s}}\mathrel{\mathop:}= \sum_{n\in\ZZ}\frac{|u_n|}{(1+|n|)^{1+ 2s}}<\infty\r\}.$$
\begin{itemize}
\item[i.] For $u\in \ell_{s}$ we have $(-\Delta)^su_n$ is given by \R{solform}, where the discrete kernel is also given by \R{kern}. 

\item[ii.] For $s\in (1,2)$ there exists constants $0 < c_s \le C_s$ such that, for any $m\in\ZZ\backslash\{0\},$ the discrete kernel $K_s$ satisfies \R{ineq1}.

\item[iii.] If $u$ is bounded, then $\lim_{s\to 1^+}(-\Delta)^s u_n = -\Delta u_n.$

\item[iv.] If $u$ is bounded, then $\lim_{s\to 2^-} (-\Delta)^s u_n = (-\Delta)^2 u_n,$ where $(-\Delta)^2$ is the classical {\em bi-harmonic operator}.
\end{itemize}
\end{theorem}

\begin{proof}
{\em i.} Let $s\in (1,2)$ and define $v_n\mathrel{\mathop:}= (-\Delta)u_n.$ Then, for $n\in\ZZ,$ we have
\bbb
(-\Delta)^s u_n & = & (-\Delta)^{s-1}(-\Delta)u_n\nnn\\
& = & (-\Delta)^{s-1}v_n\nnn\\
& = & \sum_{m\in\ZZ;\,m\neq n}(v_n - v_m)K_{s-1}(n-m)\nnn\\
& = & v_n\sum_{m\in\ZZ;\,m\neq n}K_{s-1}(n-m) - \sum_{m\in\ZZ;\,m\neq n}v_mK_{s-1}(n-m)\nnn\\
& = & A_{s-1}v_n - \sum_{m\in\ZZ;\,m\neq n}v_mK_{s-1}(n-m),\label{p1}
\eee
where 
\bb{p2}
A_s\mathrel{\mathop:}= \frac{4^s\Gamma(1/2+s)}{\sqrt{\pi}\,\Gamma(1+s)}.
\ee
Employing the definition of $v_n$ in \R{p1} yields
\bbb
(-\Delta)^s u_n & = & A_{s-1}[2u_n - u_{n-1} - u_{n+1}] - \sum_{m\in\ZZ;\,m\neq n}[2u_m - u_{m-1} - u_{m+1}]K_{s-1}(n-m)\nnn\\
& = & A_{s-1}[2u_n - u_{n-1} - u_{n+1}] - [2u_{n-1}-u_{n-2}-u_n]K_{s-1}(1)\nnn\\
&& ~~~~~ - [2u_{n+1}-u_n - u_{n+2}]K_{s-1}(-1) - \sum_{m\in\ZZ;\,m\neq n,n\pm 1}[2u_m - u_{m-1} - u_{m+1}]K_{s-1}(n-m)\nnn\\
& = & [2A_{s-1} + 2K_{s-1}(1)]u_n - \sum_{m\in\ZZ;\,m\neq 0} u_{n-m} [2K_{s-1}(m) - K_{s-1}(m-1) - K_{s-1}(m+1)]\nnn\\
& = & \gamma_{s-1}^{(1)}u_n - \sum_{m\in\ZZ;\,m\neq 0}u_{n-m}\gamma_{s-1}^{(2)}(m),\label{p3}
\eee
where 
$$\gamma_{s-1}^{(1)}\mathrel{\mathop:}= 2A_{s-1} + 2K_{s-1}(1)\quad \mbox{and}\quad \gamma_{s-1}^{(2)}(m)\mathrel{\mathop:}= 2K_{s-1}(m) - K_{s-1}(m-1) - K_{s-1}(m+1)$$
and we have used the fact that $K_{s-1}(-1) = K_{s-1}(1).$ In order to obtain the desired result, we must show that $\gamma_{s-1}^{(1)} = A_s$ and $\gamma_{s-1}^{(2)}(m) = K_s(m),\ m\in\ZZ\backslash\{0\}.$ We proceed by direct calculation. First, we note that
\bbbb
K_{s-1}(1) &=& \frac{4^{s-1}\Gamma(s-1/2)\Gamma(2-s)}{\sqrt{\pi}|\Gamma(1-s)|\Gamma(1+s)}\\
& = & \frac{4^{s-1}\Gamma(s-1/2)(1-s)\Gamma(1-s)}{\sqrt{\pi}|\Gamma(1-s)|s\Gamma(s)}\\
& = & \frac{4^{s-1}\Gamma(s-1/2)(s-1)}{\sqrt{\pi}s\Gamma(s)},
\eeee
since $s\in(1,2),$
which yields
\bbb
\gamma_{s-1}^{(1)} & = & 2A_{s-1} + 2K_{s-1}(1)\nnn\\
& = & 2\frac{4^{s-1}\Gamma(s-1/2)}{\sqrt{\pi}\Gamma(s)} + 2\l[\frac{4^{s-1}\Gamma(s-1/2)(s-1)}{\sqrt{\pi}s\Gamma(s)}\r]\nnn\\
& = & \frac{2\cdot 4^{s-1}\Gamma(s-1/2)}{\sqrt{\pi}\Gamma(s)}\l[1 + \frac{s-1}{s}\r]\nnn\\
& = & \frac{4^s\Gamma(s+1/2)}{\sqrt{\pi}\Gamma(1+s)}\nnn\\
& = & A_s.\label{p4}
\eee
In order to prove the remaining equality, we note that we can rewrite \R{kern} as
\bb{krewrite}
K_s(m) = \frac{(-1)^{m+1}\Gamma(2s+1)}{\Gamma(1+s+m)\Gamma(1+s-m)},\quad m\neq 0,
\ee
by employing the duplication and Euler reflection formula to each Gamma function, as was shown in \cite{ciaurri2015fractional}. Thus, we have
\bbb
\gamma_{s-1}^{(2)}(m) & = & 2K_{s-1}(m) - K_{s-1}(m-1) - K_{s-1}(m+1)\nnn\\
& = & (-1)^m\Gamma(2s-1)\l[\frac{-2}{\Gamma(s+m)\Gamma(s-m)}\r.\nnn\\
&& ~~~~~~~~~~ \l. - \frac{s+m-1}{\Gamma(s+m)\Gamma(s-m)(s-m)} - \frac{s-m-1}{(s+m)\Gamma(s+m)\Gamma(s-m)}\r]\nnn\\
& = & \frac{(-1)^m\Gamma(2s-1)}{\Gamma(s+m)\Gamma(s-m)}\l[-2 - \frac{s+m-1}{s-m} - \frac{s-m-1}{s+m}\r]\nnn\\
& = & \frac{(-1)^m\Gamma(2s-1)}{\Gamma(s+m+1)\Gamma(s+m-1)}\l[-4s^2 + 2s\r]\nnn\\
& = & \frac{(-1)^{m+1}\Gamma(2s+1)}{\Gamma(1+s+m)\Gamma(1+s-m)}\nnn\\
& = & K_s(m), \label{p5}
\eee
by \R{krewrite}. Combining \R{p4} with \R{p5} yields the desired result.

\vspace{3mm}

\no{\em ii.} This result follows from the the application of Lemma 9.2, from \cite{ciaurri2016nonlocal}, to \R{kern}.

\vspace{3mm}

\no{\em iii.} Following the ideas from {\em i.}, we can write
$$(-\Delta)^s u_n = P_1 + P_2,$$
where
$$P_1\mathrel{\mathop:}= (-u_{n-1}+2u_n-u_{n+1})K_s(1)\quad\mbox{and}\quad P_2\mathrel{\mathop:}= \sum_{m\in\ZZ;\,m\neq 0,1} (u_n - u_{n-m})K_s(m).$$
We obtain the desired result if we show $K_s(1) \to 1$ and $P_2 \to 0,$ as $s\to 1^+.$ To that end, we have
$$\lim_{s\to 1^+} K_s(1) = \lim_{s\to 1^+} \frac{4^s\Gamma(1/2+s)}{\sqrt{\pi}|\Gamma(-s)|}\cdot\frac{\Gamma(1-s)}{\Gamma(2+s)} = \frac{4\Gamma(3/2)}{\sqrt{\pi}\Gamma(3)} = 1.$$
Now, by the assumption that $u$ is bounded---that is, $\|u\|_{\ell^\infty} < \infty,$ where $\|\cdot\|_{\ell^\infty}$ is the norm  on $\ell^\infty(\ZZ)$---we have
$$\|P_2\|_{\ell^\infty} \le 2\|u\|_{\ell^\infty}\sum_{m\in\ZZ;\,m\neq 0,1} K_s(m) = 2\|u\|_{\ell^\infty}\l[\frac{4^s\Gamma(1/2+s)}{\sqrt{\pi}\Gamma(1+s)} - 1\r].$$
Since $\textstyle\lim_{s\to 1^+} 4^s\pi^{-1/2}\Gamma(1/2+s)/\Gamma(1+s) = 1,$ we have 
$$\|P_2\|_{\ell^\infty} \to 0, \quad \mbox{as}\ s\to 1^+,$$
which is the desired result.

\vspace{3mm}

\no{\em iv.} We begin by recalling that the discrete biharmonic operator is given by
\bb{biharm}
(-\Delta)^2u_n = u_{n-2} - 4u_{n-1} + 6u_n - 4u_{n+1} + u_{n+2}.
\ee
Just as before, by symmetry we can write
$$(-\Delta)^su_n = S_1 + S_2 + S_3,$$
where
$$S_1 \mathrel{\mathop:}= (-u_{n-1}+2u_n-u_{n+1})K_s(1),\quad S_2 \mathrel{\mathop:}= (-u_{n-2} + 2u_n - u_{n+2})K_s(2),$$
and
$$S_3 \mathrel{\mathop:}= \sum_{m\in\ZZ;\,m\neq 0,1,2} (u_n - u_{n-m})K_s(m).$$
Similar to before, we show that $K_s(1) \to 4,$ $K_s(2) \to -1,$ and $S_3 \to 0,$ as $s\to 2^-,$ where
$$\lim_{s\to 2^-}K_s(1) = \lim_{s\to 2^-}\frac{4^s\Gamma(1/2+s)}{\sqrt{\pi}|\Gamma(-s)|}\cdot\frac{\Gamma(1-s)}{\Gamma(2+s)} = \frac{32\Gamma(5/2)}{\sqrt{\pi}\Gamma(4)} = 4$$
and
$$\lim_{s\to 2^-}K_s(2) = \lim_{s\to 2^-}\frac{4^s\Gamma(1/2+s)}{\sqrt{\pi}|\Gamma(-s)|}\cdot\frac{\Gamma(2-s)}{\Gamma(3+s)} = \frac{-16\Gamma(5/2)}{\sqrt{\pi}\Gamma(5)} = -1.$$
Once again, by the assumption that $u$ is bounded, we have
$$\|S_3\|_{\ell^\infty} \le 2\|u\|_{\ell^\infty} \sum_{m\in\ZZ;\,m\neq 0,1,2}K_s(m) = 2\|u\|_{\ell^\infty}\l[\frac{4^s\Gamma(1/2+s)}{\sqrt{\pi}\Gamma(1+s)} - 4 + 1\r].$$
Since $\textstyle\lim_{s\to 2^-} 4^s\pi^{-1/2}\Gamma(1/2+s)/\Gamma(1+s) = 3,$ we have 
$$\|S_3\|_{\ell^\infty} \to 0, \quad \mbox{as}\ s\to 2^-,$$
which is the desired result.
\end{proof}

Thus, we may use \R{solform} and \R{kern} for numerical approximations of the random discrete fractional Schr{\"o}dinger operator given in \R{schro1}. Finally, in order to provide the readers with a more concrete understanding of the discrete weights studied in this section, we include a plot of $K_s(m),$ for various values of $s\in (0,2).$ Of particular interest is the smooth transition between the weight functions in the regimes $s\in(0,1)$ and $s\in(1,2),$ while there is an abrupt qualitative shift at $s=1.$ Moreover, one is able to clearly see the rapid decay of the values of $K_s(m),$ as $m\to\pm\infty.$

\begin{figure}[H]
\centering
\includegraphics[width=2.63in,height=1.68in]{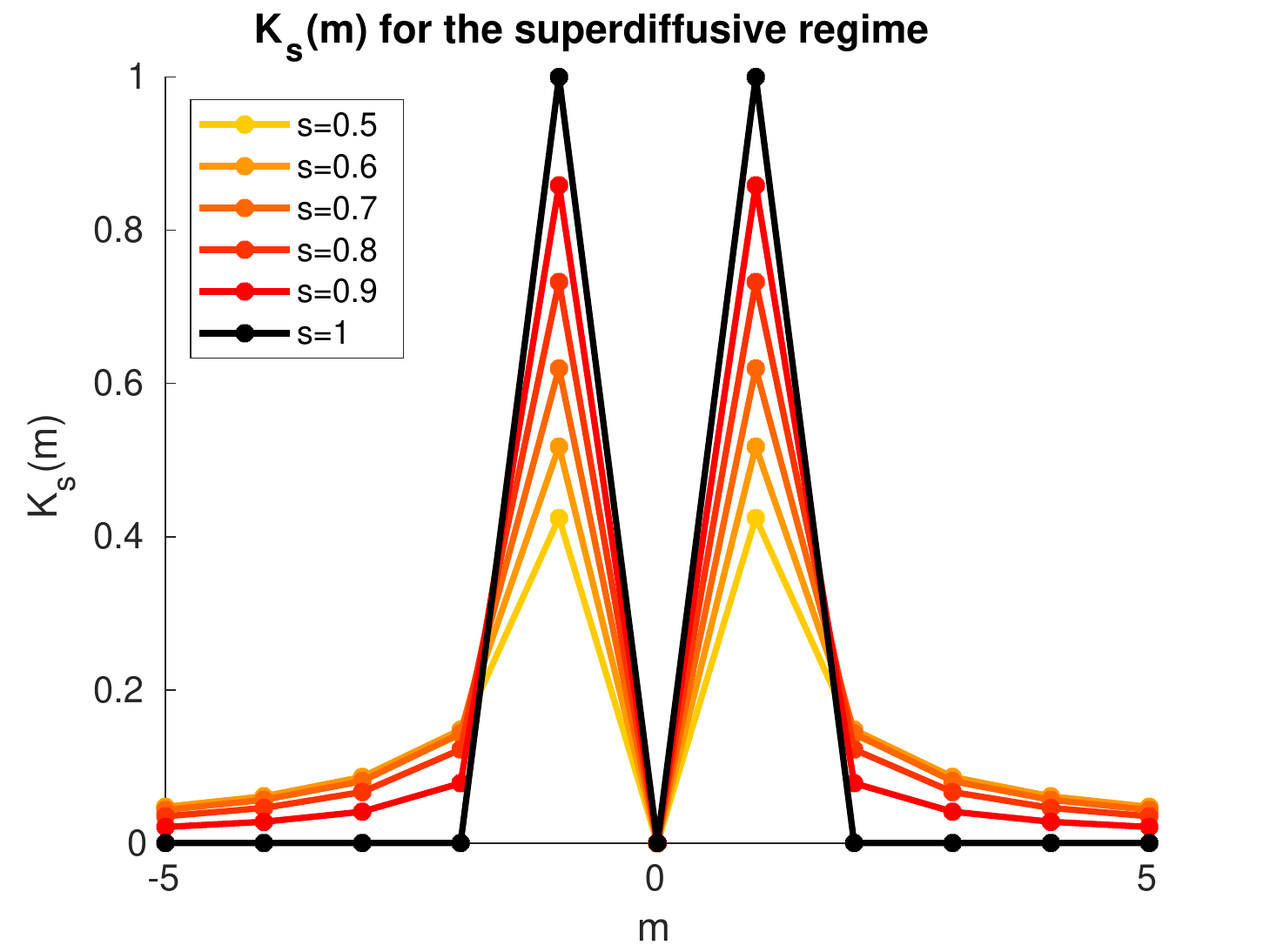}
\includegraphics[width=2.63in,height=1.68in]{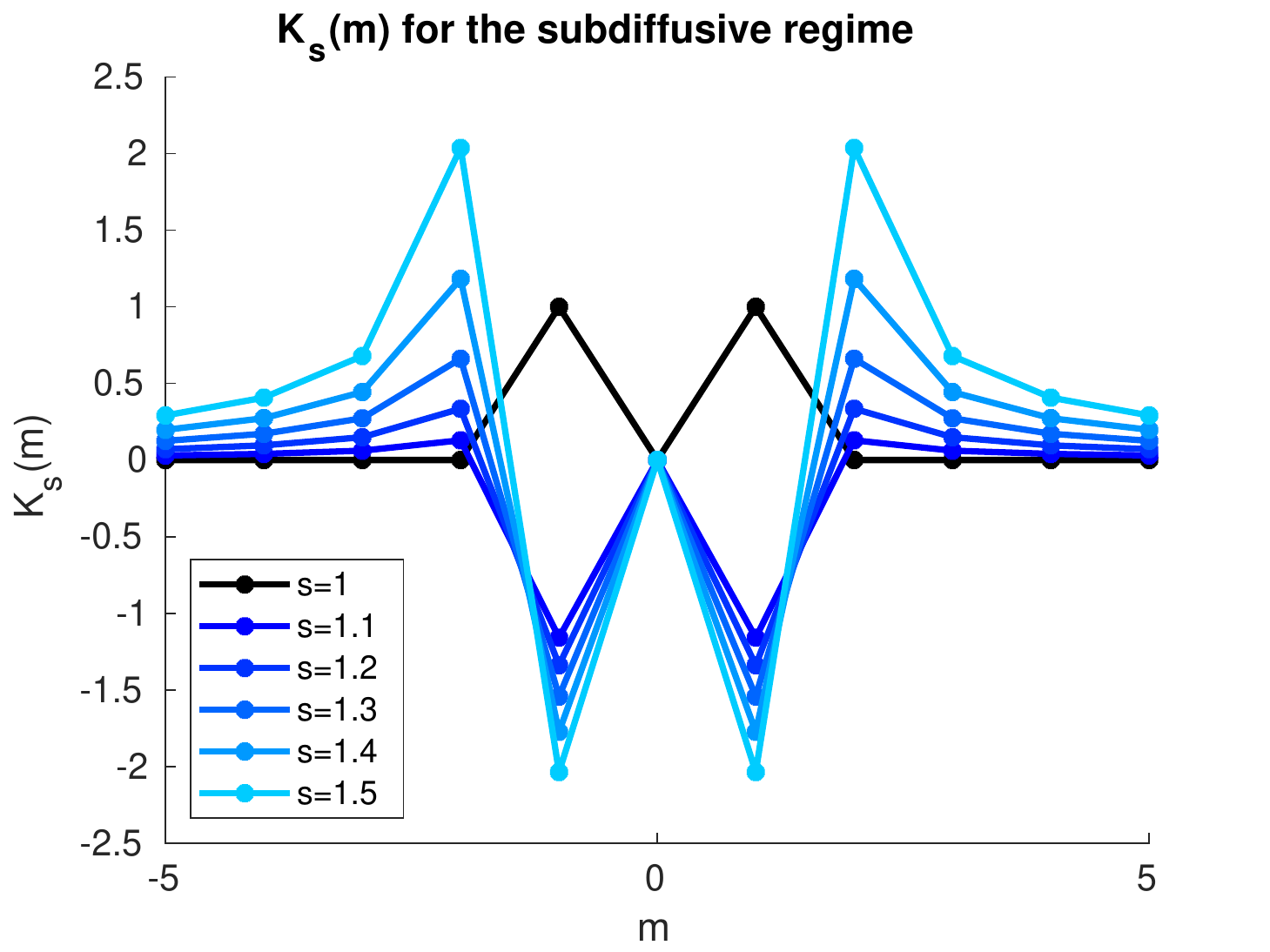}
\caption{Plots of the discrete weight functions for various $s-$values. [LEFT] A plot of the discrete weights corresponding to the superdiffusive parameter regime. [RIGHT] A plot of the discrete weights corresponding to the subdiffusive parameter regime. The monotonic transition within each parameter regime is evident, with a clear qualitative transition occurring as $s\to 1^+.$ 
}
\end{figure}

\subsection{Physical Interpretation}

A physical interpretation of the discrete fractional Laplacian has been provided in \cite{ciaurri2016nonlocal}. We include that description for completeness and expand it to include our newly developed model. Let $u$ be a discrete harmonic function on $\ZZ,$ that is, $-\Delta u = 0.$ Then $u$ also satisfies the following discrete mean value property:
\bb{dis1}
u_n = \frac{1}{2}u_{n-1}+\frac{1}{2}u_{n+1}.
\ee
This classical result provides the physical intuition that is well understood for the classical discrete Laplacian in one-dimension. That is, by \R{dis1}, we see that a discrete harmonic function represents a physical situation in which a particle will jump to either of the two adjacent nodes with equal probability, namely, one-half. This intuition may be generalized to the current situation to provide a physical interpretation of the fractional Laplacian. 

Assume that $u$ is a fractional discrete harmonic function, that is, $(-\Delta)^su = 0.$ Then by \R{solform} we have that $u$ satisfies the following discrete fractional mean value property:
\bb{dis2}
u_n = \sum_{m\in\ZZ,m\neq n}u_mP_s(n-m),
\ee
where
\bb{disp}
P_s(m)\mathrel{\mathop:}= \frac{1}{A_s}K_s(m),
\ee
$A_s \mathrel{\mathop:}= \textstyle\sum_{m\in\ZZ} K_s(m),$ and $P_s(0) = 0.$ In this representation, $P_s(m)$ is a probability distribution on $\ZZ,$ allowing one to interpret the fractional case in a similar fashion to the classical result given by \R{dis1}. That is, by \R{dis2}, it follows that a fractional discrete harmonic function describes a particle which may jump to {\em any} point in $\ZZ$ and the probability that the particle jumps from point $n$ to point $m$ is given by $P_s(n-m).$ If $s\in (0,1),$ Theorem 1 implies that the probability of jumping from point $n$ to point $m$ is proportional to $|n-m|^{-(1+2s)}.$ In this situation, as $s\to 1^-,$ the probability of jumping from $n$ to an adjacent point tends to one, while the probability of jumping to a nonadjacent point tends to zero. Further, as $s\to 0^+,$ the probability of jumping from the point $n$ to any point in $\ZZ$ tends to zero, resulting in no jumps, as a value of $s=0$ represents no diffusion.

Theorem 2 also allows for a physical interpretation to hold for $s\in (1,2),$ although the construction of this interpretation requires more care. In the case that $s\in (1,2),$ we may write $s = 1 + \tilde{s},$ where $\tilde{s}\in(0,1).$ We can then decompose our operator as
\bb{factored_op}
(-\Delta)^su_n = (-\Delta)^{\tilde{s}}(-\Delta)u_n
\ee
and perform the computation in stages. First, we apply the standard Laplace operator to $u_n$ to obtain
\bb{factored_v}
v_n \mathrel{\mathop:}= (-\Delta)u_n,
\ee
whose interpretation is exactly analagous to the classical setting: that is, a particle located at position $n$ will jump to either $n-1$ or $n+1$ with equal probability. We then apply the nonlocal operator to $v_n;$ that is, we compute $(-\Delta)^{\tilde{s}}v_n.$ This action has the same probabilistic interpretation from the previous paragraph (see \R{dis2} and \R{disp}). 
However, when the two
actions are composed together, there is the possibility of the particle moving back to its original location (thus reducing the probability that it will land at $n-1$ or $n+1$)
which is demonstrated by the possibility of negative weights (see Figure 1).

It is important to note that the negative weights does not translate into a negative probabilities. Instead, the negative weights assigned to the nearest neighbors should be interpreted as nonzero probabilities of the particle staying at its original location.
Moreover, when $s\in(1,2),$ it is the case that these negative coefficients will only occur at the ``nearest neighbor'' sites. As before, when $s\to 1^+,$ the probability of jumping from $n$ to an adjacent point tends to one, while the probability of jumping to a nonadjacent point tends to zero. Now, as $s\to 2^-,$ we have that the probabilistic interpretation converges to exactly that of the discrete biharmonic operator \cite{vanderbei1984probabilistic,mazzucchi2013probabilistic}.

\subsection{Nonlinear Mean Squared Displacement}

We now provide a more explicit description of the relationship between the fractional power, $s,$ in $H_{s,\epsilon},$ and the nonlinear mean squared displacement, $\beta,$ mentioned in the Introduction. In order to demonstrate this relationship, we consider the following non-local Cauchy problem
\bb{non1}
\l\{\begin{array}{rcll}
v_{t}(x,t)  & = & -(-\Delta)^s v(x,t) , & x\in\ZZ,\ t > 0,\\
v(x,0) & = & \varphi(x), & x\in \ZZ.
\end{array}\r.
\ee
It was shown in \cite{CIAURRI2015119} that the solution to \R{non1}, for appropriate bounded initial values, is given by
\bb{nonsol}
v(x,t) = \sum_{k\in\ZZ} G^s(x-k,t)\varphi(k),
\ee
where
\bb{W}
G^s(x,t) \mathrel{\mathop:}= \frac{1}{2\pi}\int_{-\pi}^\pi e^{t(4\sin^2(z/2))^s}e^{-ixz}\,dz,
\ee
when $s\in(0,1).$ It is relatively straightforward to generalize this result to obtain a representation for $s\in(0,2).$

\begin{theorem}
The solution to \R{non1} is given by \R{nonsol} when $s\in (1,2).$
\end{theorem}

We omit the proof of Theorem 3, for brevity, but note that the approach is similar to our methods employed in the proof of Theorem 2, combined with the methods employed in \cite{CIAURRI2015119}. Theorem 3 allows for the representation of all parameter regimes of interest via the solution form \R{nonsol}. It is then clear that $G^s(x,t)$ acts as a discrete Green's function for the problem \R{non1}. With the above representation in hand, one may calculate the mean squared displacement to be
\bb{msd}
\langle (x-x_0)^2 \rangle \sim t^{1/s},
\ee
where $s\in (0,2).$ We want to make a few notes regarding the expression in \R{msd}. We have omitted the calculations required to obtain \R{msd} as they are similar to that of the classical case, except one must employ Mittag-Leffler and Fox H-functions. Note that the contributions from the random coefficients in \R{schro1} do not contribute to the mean squared displacement calculated by \R{msd} due to the fact that we compute an expected value as part of the calculation. Also note that \R{msd} is an asymptotic relationship at $t\to\infty.$ Using this relation, we see that the mean squared displacement scales inversely with the fractional power of the Laplace operator, and that we obtain the classical linear mean squared displacement as $s\to 1.$ 


\section{Spectral Approach to Transport in Disordered Systems}

We have successfully employed the spectral approach in numerous settings and have demonstrated it to be quite effective in the study of transport behavior -- see, for instance, \cite{Liaw2013,1751-8121-47-30-305202,2053-1591-3-12-125904,PhysRevB.96.235408,kostadinova2018transport}. In this section we provide a brief overview of the spectral method employed in \cite{Liaw2013}, as it applies to the discrete fractional Schr{\"o}dinger operator.

On $\ell^2(\ZZ)$ -- the space of two-sided  square-summable sequences -- we consider the random discrete fractional Schr{\"o}dinger operator defined in \R{schro1}. That is, we consider
$$
H_{s,\epsilon} \mathrel{\mathop:}= (-\Delta)^s + \sum_{i\in\ZZ}\epsilon_i\,\langle \,\cdot\,,\delta_i\rangle \delta_i
$$
where $\langle \cdot, \cdot\rangle$ denotes the $\ell^2 (\ZZ)$ inner product, $\delta_i$ is the standard basis of $\ZZ$, and $\epsilon_i$ are random variables taken to be i.i.d. according to the uniform distribution on $[-c/2,c/2]$. We are most interested in the change in transport behavior as we vary the diffusion parameter $s.$

A vector $\varphi$ is cyclic for a (bounded) operator $T$ on a separable Hilbert space $X$, if the forward orbit of $\varphi$ under $T$ has dense linear span ({\em i.e.}, $X = \text{clos}\,\text{span}\{T^n \varphi:n\in \NN\cup\{0\}\}$).
The central result (see Corollary 3.2 of \cite{Liaw2013}) behind the spectral method can be formulated as follows:
\begin{quote}
{\em If we can find a (non-trivial) vector that is not cyclic for $(H_{s,\epsilon})$ with positive probability, then almost surely there are de-localized states.}
\end{quote}
As a side note, we mention that the existence of de-localized states indicates transport by the RAGE theorem, see e.g., Section 1.2 of \cite{MR2603225}.

The non-cyclicity of a vector follows if its forward orbit stays away from a particular direction. In other words, if we can find a vector $v$ that remains at a positive distance from $\text{span}\{H_{s,\epsilon}^k \varphi:k\in \NN, k\le n\}$ as $n\to\infty,$ then the vector is non-cyclic. With some linear algebra (see Proposition 3.1 of \cite{Liaw2013}) this distance can be expressed explicitly by
\bb{distance}
D_{s,\epsilon}^n\mathrel{\mathop:}=\sqrt{1-\sum_{k=0}^n\frac{\langle v, m_k\rangle^2}{\langle m_k, m_k\rangle}}\,,
\ee
where $\{m_k\}$ is the orthogonal sequence of $\ell^2(\ZZ)$ vectors obtained from applying the Gram--Schmidt algorithm to $\{\varphi, H_{s,\epsilon} \varphi, H^2_{s,\epsilon} \varphi, \hdots\}$.

In the current study, we choose $v$ to be a linear combination of basis vectors, in order to account for the nonlocality of the action of $H_{s,\epsilon}.$ However, it is worth noting that spectral theory allows for any $v\in\ell^2(\ZZ)$ to be an appropriate choice. With this choice, we investigate the dependence of (de)localization on the diffusion parameter $s$. To emphasize this point, we will write $ D_{s,\epsilon}^n$. 
Summarizing the theory, our numerical investigations employ the following tool:
\[
\lim_{n\to\infty}D_{s,\epsilon}^n > 0
\qquad\Rightarrow\qquad\text{de-localization}.
\]

We also note that the operator $H_{s,\epsilon}$ is self-adjoint in $\ell^2(\ZZ).$ This follows immediately from the spectral theorem and the fact that the discrete Laplacian in self-adjoint in $\ell^2(\ZZ).$ This fact allows us to apply the efficient computational techniques outlined in \cite{Liaw2013}. We omit these details for brevity, but use them when performing the numerical experiments in the ensuing section.

\section{Numerical Experiments}





This section outlines the numerical method and provides numerical simulations which verify the proposed approach. In Section 4.1, the formulation of the fractional Laplacian given in Theorems 1 and 2 is used to justify the numerical method. Section 4.2 then provides relevant simulation results for the method obtained for various parameter choices. Finally, an orthogonality check to confirm that the employed forward Gram--Schmidt algorithm creates orthogonal vectors in this setting is outlined and performed in Section 4.3. 

It is worth noting that the experiments presented do not provide rigorous justification or proof for the occurrence of localized or extended states. Rather, they provide a qualitative analysis of the transport behavior of the modeled systems as the fractional power $s$ is varied. A more detailed study is planned in forthcoming work, with initial efforts being focused on an expanded numerical study with accompanying physical and mathematical interpretations, including the scope and limitations of the new technique. Moreover, we will focus on applying our results to specific physical problems.

\subsection{Motivation of Computational Method}



As outlined in Section 3.1, the spectral approach employed in this study requires the examination of the forward orbit of arbitrarily chosen initial vectors under the action of the random discrete fractional Schr{\"o}dinger operator. Thus, our computational method must accurately and efficiently apply the operator given by \R{schro1}. The case $s=1$ is relatively straightforward to implement due to the definition of the discrete Laplacian given in \R{dlap}. This definition will be used as a special case of our current approach and has been considered in more detail in \cite{Liaw2013,PhysRevB.96.235408,kostadinova2018transport}. 

We consider the one-dimensional lattice $\ZZ$ and an arbitrary function $u\,:\,\ZZ\to \RR.$ Then by \R{solform}, we have
$$(-\Delta)^s u_n = \sum_{m\in\ZZ;\,m\neq n} (u_n-u_m)K_s(n-m),$$
for $s\in(0,2)$ and the kernel $K_s$ given by \R{kern}. Due to the symmetry of the kernel, we may rewrite the above and obtain
\bb{comp1}
(-\Delta)^s u_n = \sum_{m\in\NN} (2u_n-u_{n-m}-u_{n+m})K_s(m)
\ee
for $s\in(0,2).$ Thus, for some $1 \ll M\in\NN,$ we have
\bb{comp2}
(-\Delta)^s u_n = \sum_{m=1}^M (2u_n-u_{n-m}-u_{n+m})K_s(m) + R_M(u_n),
\ee
where
\bb{comp3}
R_M(u_n) \mathrel{\mathop:}= \sum_{m=M+1}^\infty (2u_n-u_{n-m}-u_{n+m})K_s(m).
\ee

For simplicity, our computational method will disregard the remainder term, $R_M(u_n).$ A similar truncation was employed for the simulations in \cite{ciaurri2015fractional,ciaurri2016nonlocal}, as well. The remainder is guaranteed to be bounded and well-controlled by our choice of $M$ due to \R{ineq1}. In fact, we have
\bbbb
|R_M(u_n)| & \le & 4\max_{m>M}|u_m|\int_M^\infty K_s(x)\,dx\\
& = & B_s \int_M^\infty \frac{\Gamma(x-s)}{\Gamma(x+1+s)}\,dx\\
& = & \tilde{B}_s \int_M^\infty \int_0^\infty e^{-(x-s)y}(1-e^{-y})^{2s}\,dy\,dx\\
& \le & \tilde{B}_s\int_M^\infty\int_0^\infty e^{-(x-s)y}\,dy\,dx\\
& = & \frac{\tilde{B}_s}{(M-s)^2},
\eeee
where 
$$B_s \mathrel{\mathop:}= 4\max_{m>M}|u_m|\cdot\frac{4^s\Gamma(1/2+s)}{\sqrt{\pi}|\Gamma(-s)|} \quad \mbox{and}\quad \tilde{B}_s\mathrel{\mathop:}= \frac{B_s}{\Gamma(1+2s)},$$ 
and we have employed the property
\bb{gammaprop}
\frac{\Gamma(x-s)}{\Gamma(x+1+s)} = \frac{1}{\Gamma(1+2s)}\int_0^\infty e^{-(x-s)y}(1-e^{-y})^{2s}\,dy,
\ee
which is valid for $x-s>0$ \cite{LAFORGIA2012833}. Moreover, we have that
$$\max_{s\in(0,2)} \tilde{B}_s \approx 1.27324\times\max_{m>M}|u_m|,$$
occurring when $s = 1/2,3/2,$ which yields
\bb{r_m}
R_M(u_m) \sim \frac{1}{M^2}.
\ee

Thus, dropping the term $R_M(u_n)$ in \R{comp2} appears reasonable. By doing so, we obtain
\bb{comp5}
H_{s,\epsilon}u_n \approx \sum_{m=1}^M (2u_n - u_{n-m}-u_{n+m})K_s(m) + \epsilon_nu_n,
\ee
where the $\epsilon_n$ are i.i.d.~according to the uniform distribution on $[-c/2,c/2],$ for some fixed $c>0.$ The approximation given by \R{comp5} is employed in the following computations. Since the goal of the current project is to explore the behavior of \R{schro1}, we leave detailed error and convergence analysis of our approximation given by \R{comp5} for future work.

\subsection{Numerical Simulations in a One-Dimensional Disordered System}



Consider the discrete random fractional Schr{\"o}dinger operator given by \R{schro1} with independently and identically distributed random variables $\epsilon_i.$ The spectral approach outlined in Section 3 dictates that if we can find a disorder $c>0$ for which 
\bb{num1}
D_{s,\epsilon}\mathrel{\mathop:}= \lim_{n\to\infty} D^n_{s,\epsilon} > 0,
\ee
with nonzero probability, then \R{schro1} will exhibit de-localized energy states. We will now explain how one can verify de-localization numerically.

In the numerical experiments, we initially fix $c$ and fix one computer-generated realization of the random variables $\epsilon_i.$ In our case, these random variables are uniformly distributed in $[-c/2,c/2].$ We then calculate the values of $D_{s,\epsilon}^n$ for $n\in\{0,1,2,...\}$ and each $s-$value of interest. Since $D_{s,\epsilon}^n$ is a positive, monotonically decreasing sequence, one can construct approximate lower bounds (with respect to the probability distribution) of the limit $D_{s,\epsilon}.$ This is exactly the approach taken in \cite{PhysRevB.99.024115,Liaw2013,1751-8121-47-30-305202,2053-1591-3-12-125904,PhysRevB.96.235408}. In these works, it is noted that the distance values may decay logarithmically, so the authors performed a careful re-scaling of the horizontal axis by an exponent $a<0.$ A lower bound for all possible $y-$intercepts of any linear approximation of the re-scaled distance values is then computed, and serves as a lower bound for $D_{1,\epsilon}.$ A positive lower bound implies the existence of extended states, with the results being quite strong due to the overestimation of the likely lower bound. In fact, it was shown there that de-localized states exist for nontrivial disorder values in various two-dimensional geometries. This being the first exploration of its kind, our primary goal is to demonstrate the qualitative difference between the anomalous and classical diffusion cases. To aid in the interpretation of the qualitative differences in the plots, we provide the following criterion: 

\begin{quote}
{\bf Numerical Interpretation Criterion:} For a fixed realizations of $\epsilon_i,$ a fixed vector $v,$ and the integer $n$ sufficiently large, if $D_{s_2,\epsilon}^n > D_{s_1,\epsilon}^n$ and $H_{s_1,\epsilon}$ exhibits de-localized states, then the $H_{s_2,\epsilon}$ exhibits de-localized states. The converse of this statement holds, as well.
\end{quote}

Since it is known that all energy states for \R{schro1} will be localized when $s=1,$ it follows that $D_{1,\epsilon} = 0,$ for all choices of the vector $v.$ This important result provides a baseline, for various non-zero $c-$values, to be combined with the aforementioned Numerical Interpretation Criterion. Any distance plots decaying more slowly than those for $s=1$ can be interpreted as exhibiting {\em enhanced transport behavior}, as compared to the known localized behavior of $s=1.$ Similarly, any distance plot decaying more quickly than those for $s=1$ can be interpreted as exhibiting {\em inhibited transport behavior}, as compared to the known localized behavior of $s=1.$ However, it is important to note that such occurrences do not guarantee the existence (or lack thereof) of de-localized states.



As an example, we compare the distance plots for $s=0.9,1.0,1.1,$ for several different disorder values. The values $s=0.9$ and $s=1.1$ are arbitrarily chosen superdiffusive and subdiffusive parameters, respectively, while $s=1$ corresponds to the classical case. In Figure 2, we have fixed $M = 300$ (as in \R{comp5}) and 
\bb{vchoice}
v = \sum_{i=0}^{M-1} \l[(-1)^i\delta_{1+i} + (-1)^{i+1}\delta_{-1-i}\r]
\ee
(as in \R{distance}), for each simulation. More details regarding the choice of $v$ will be explored in future work. However, we note that this choice of $v$ is motivated by the need to consider the nonlocal effects of the underlying operator. The parameter regime for $c$ has been purposely chosen small, due to the expected dominance of localization behavior in one-dimensional systems. Each plot is the result of averaging ten simulations, in an effort to reduce the spurious effects of certain realizations of the random perturbations.


\begin{figure}[H]
\centering
\includegraphics[width=2.63in,height=1.68in]{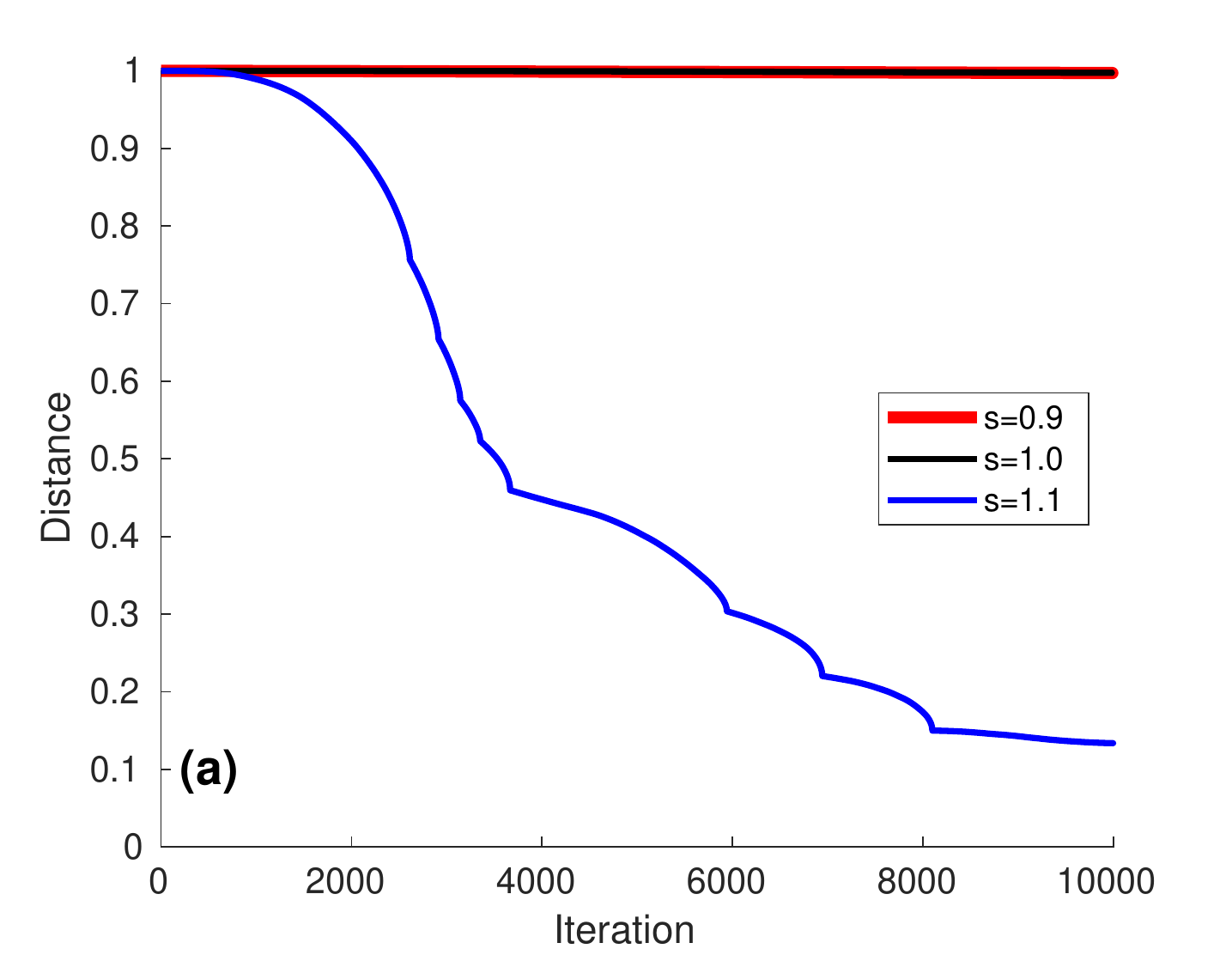}
\includegraphics[width=2.63in,height=1.68in]{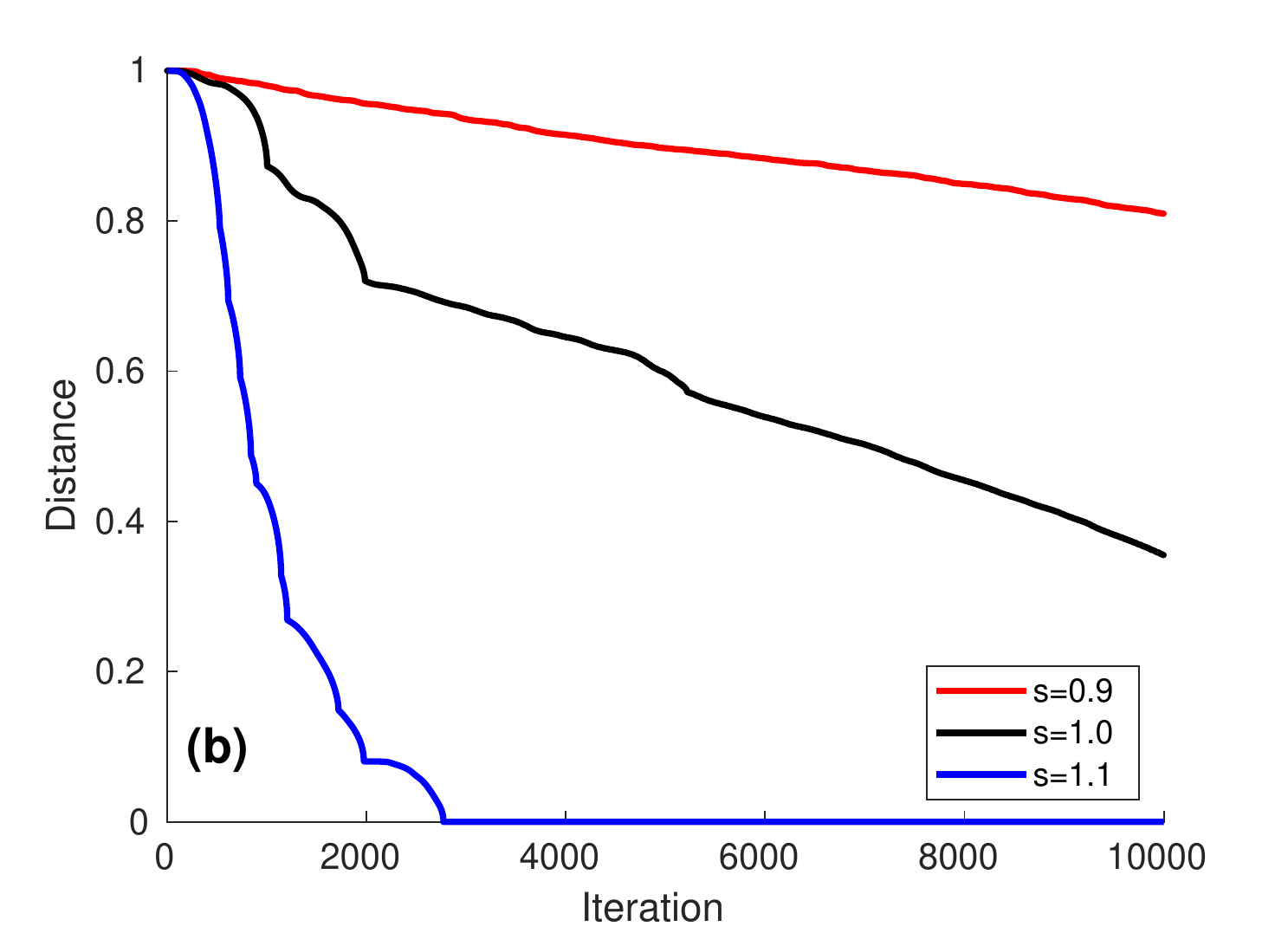}
\vspace{5mm}
\includegraphics[width=2.63in,height=1.68in]{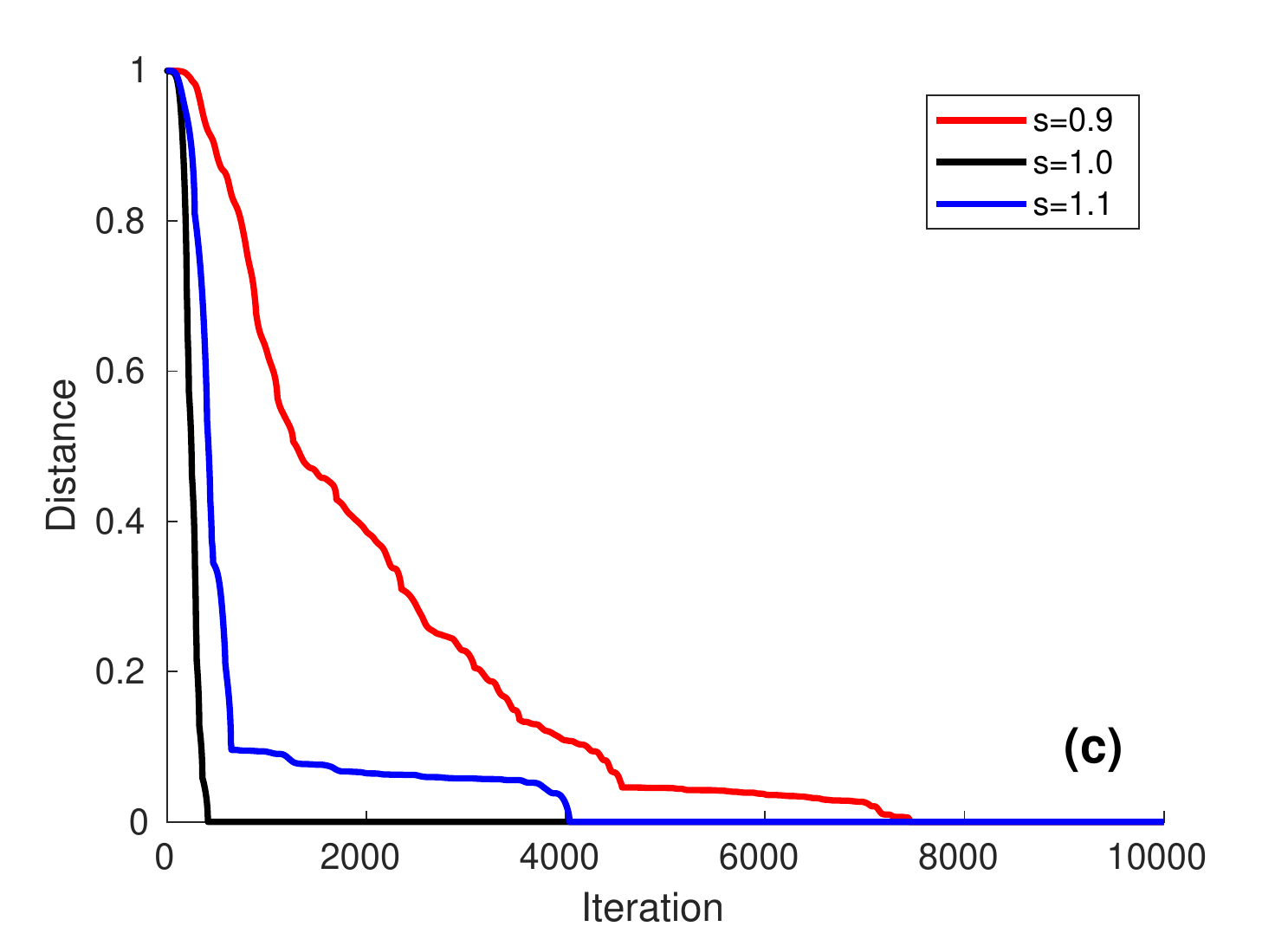}
\includegraphics[width=2.63in,height=1.68in]{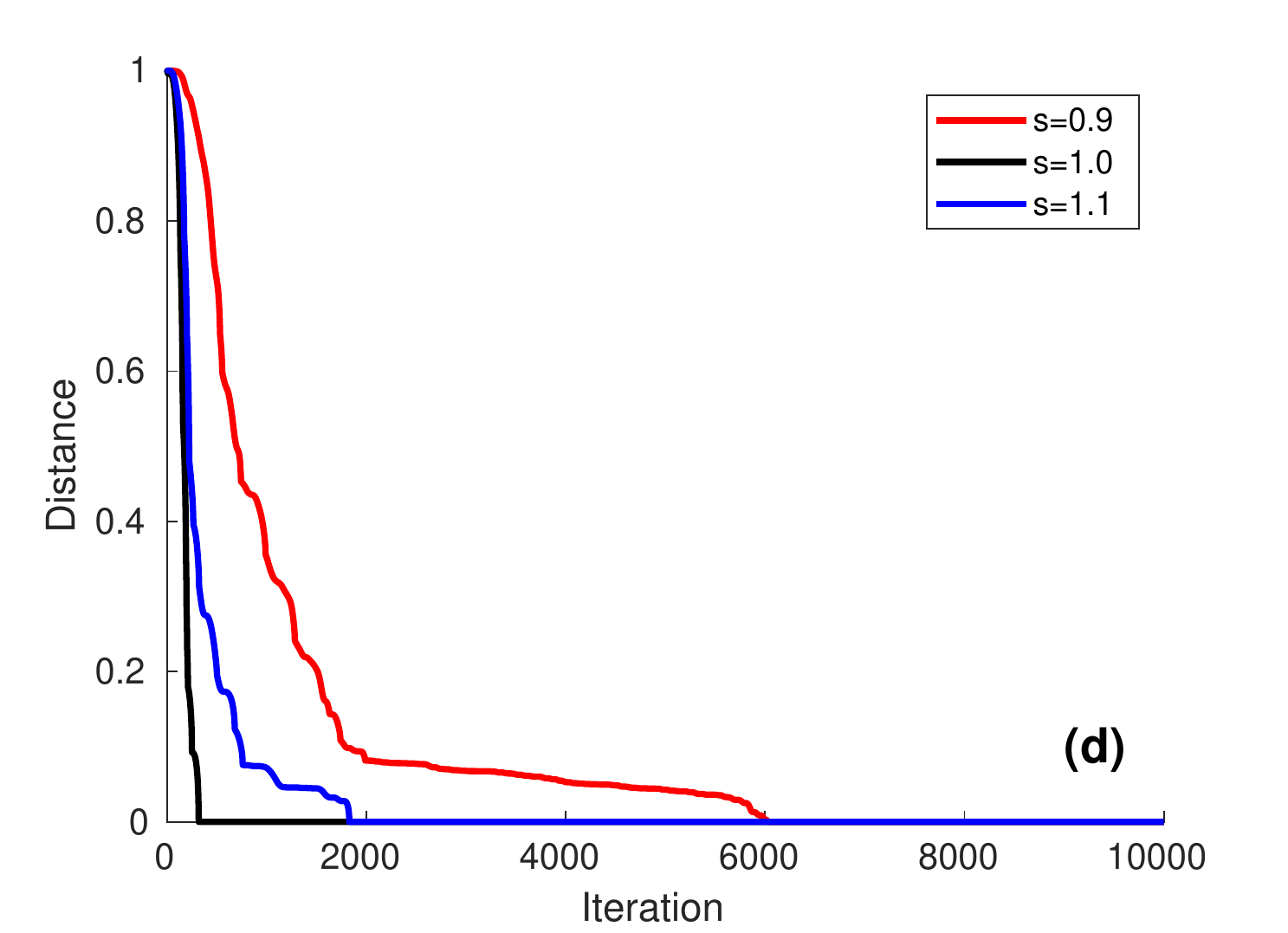}
\caption{Plots of the function $D_{s,\epsilon}^n$ for various parameter regimes (with $n$ being the number of iterations). Each plot considers the cases $s=0.9,1,1.1.$ The $c-$values considered are (a) $c = 0.0001,$ (b) $c = 0.001,$ (c) $c= 0.005,$ and (d) $c=0.01.$ The graphs of $D_{s,\epsilon}^n$ decay more quickly as $c$ increases, with the graph corresponding to $s=0.9$ decaying more slowly than the graphs of $s=1$ and $s=1.1,$ in each case.}
\end{figure}

From Figure 2, it is clear that the superdiffusive case, $s=0.9,$ results in a distance plot that decays more slowly, with respect to the number of time steps, than the other two cases in each of the presented plots. We see that the subdiffusive graphs decay more quickly than the classical case for the smallest values of disorder, $c,$ but then does so more slowly for $c\ge 0.005.$ 
We remark that the first two plots are the expected results, while these latter two plots could be due to random variations of the underlying probability realizations. 
This slower decay as $c$ increases is not surprising, as the subdiffusive case represents nonlocal interactions and the $s=1$ case represents local (nearest neighbor) interactions. The nonlocality of the operator when $s=0.9$ contributes to this slower decay due to the nonzero probability of a long-range jump.
The complicated nature of these nonlocal, subdiffusive interactions were outlined in Section 2.2. 
Regardless of the $c-$value, we see that both the subdiffusive and classical cases have a very steep initial decent, while the superdiffusive case descends more slowly. 

These plots demonstrate that the transport behavior, especially for small disorder, can be quite different. It is interesting to point out the notably slower decay of the superdiffusive case, as compared to the classical case, which is known to localize \cite{MR2603225}. These observations raise the question as to whether de-localized energy states can be observed when $s\in(0,1).$ These results will be explored both computationally and theoretically in more detail in the future.

These behaviors are still preserved for a different choice of $v,$ as demonstrated in Figure 3. 
We still consider vectors of the form \R{vchoice}, however, we consider the two distinct vectors: one corresponding to $M=300$ and one corresponding to $M=400.$
As before, there is an expectation of slight differences in the results due to fluctuations in the realizations of the random parameters. However, we still see a strong agreement in qualitative behavior. Figures 3a and 3b represent results for two different $c-$values, but the same fixed vector $v$ with $M=300.$ Figures 3c and 3d show results for the same $c-$values as Figures 3a and 3b, but compute the distance with a vector $v$ with $M=400.$ 
We see that even with different choices of the vector $v,$ the plots agree reasonably well. 


\begin{figure}[H]
\centering
\includegraphics[width=2.63in,height=1.68in]{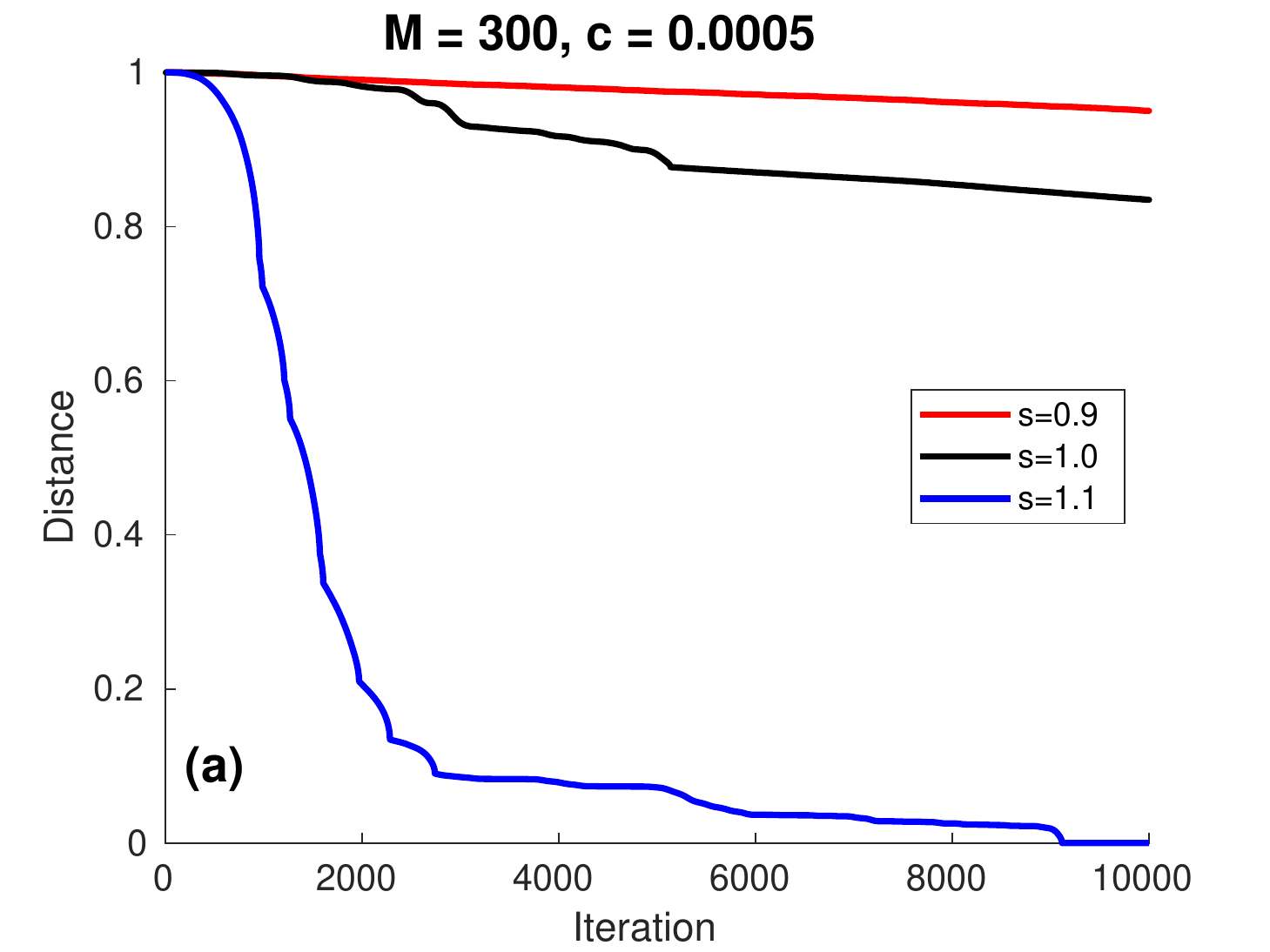}
\includegraphics[width=2.63in,height=1.68in]{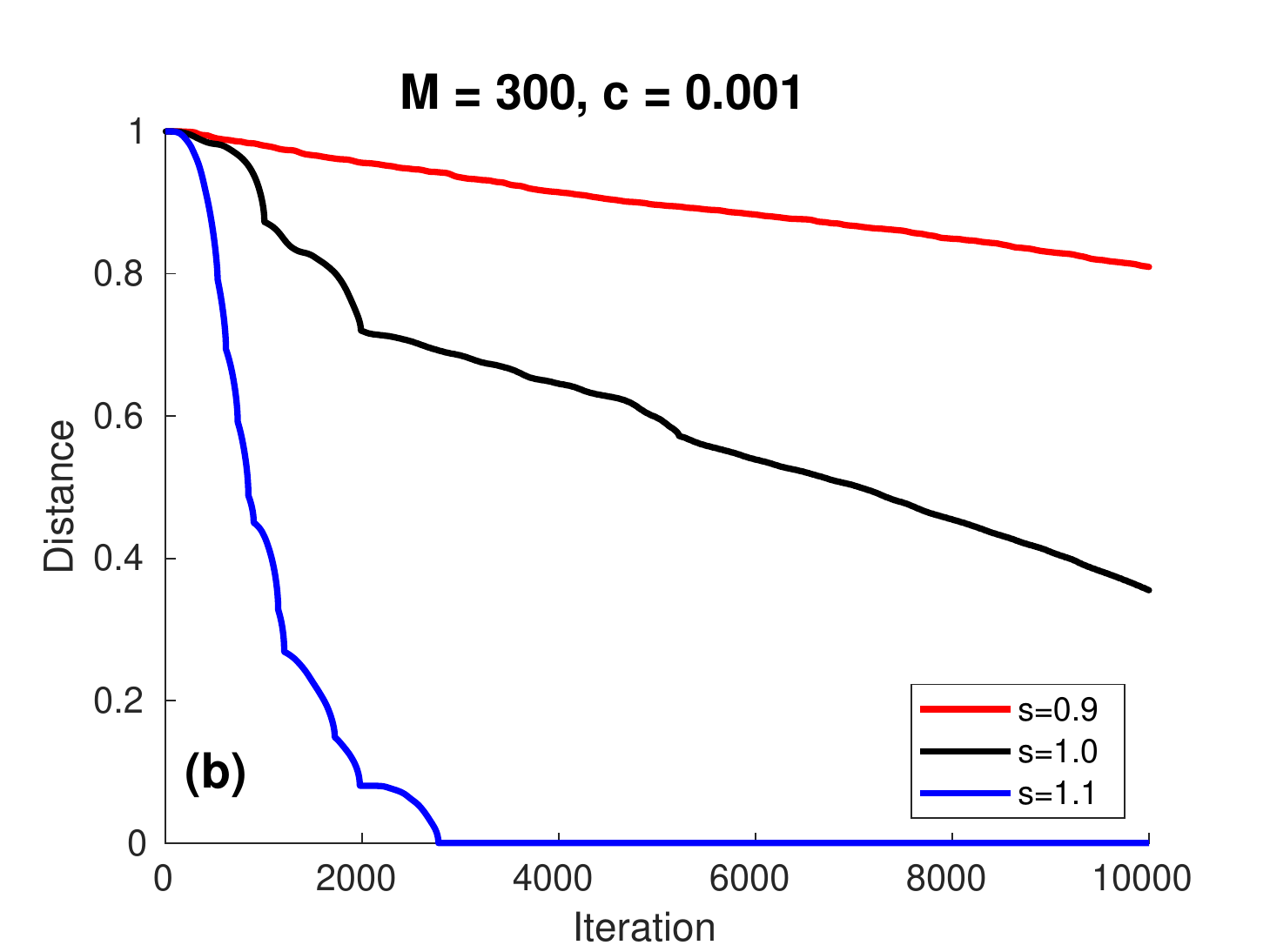}
\vspace{5mm}
\includegraphics[width=2.63in,height=1.68in]{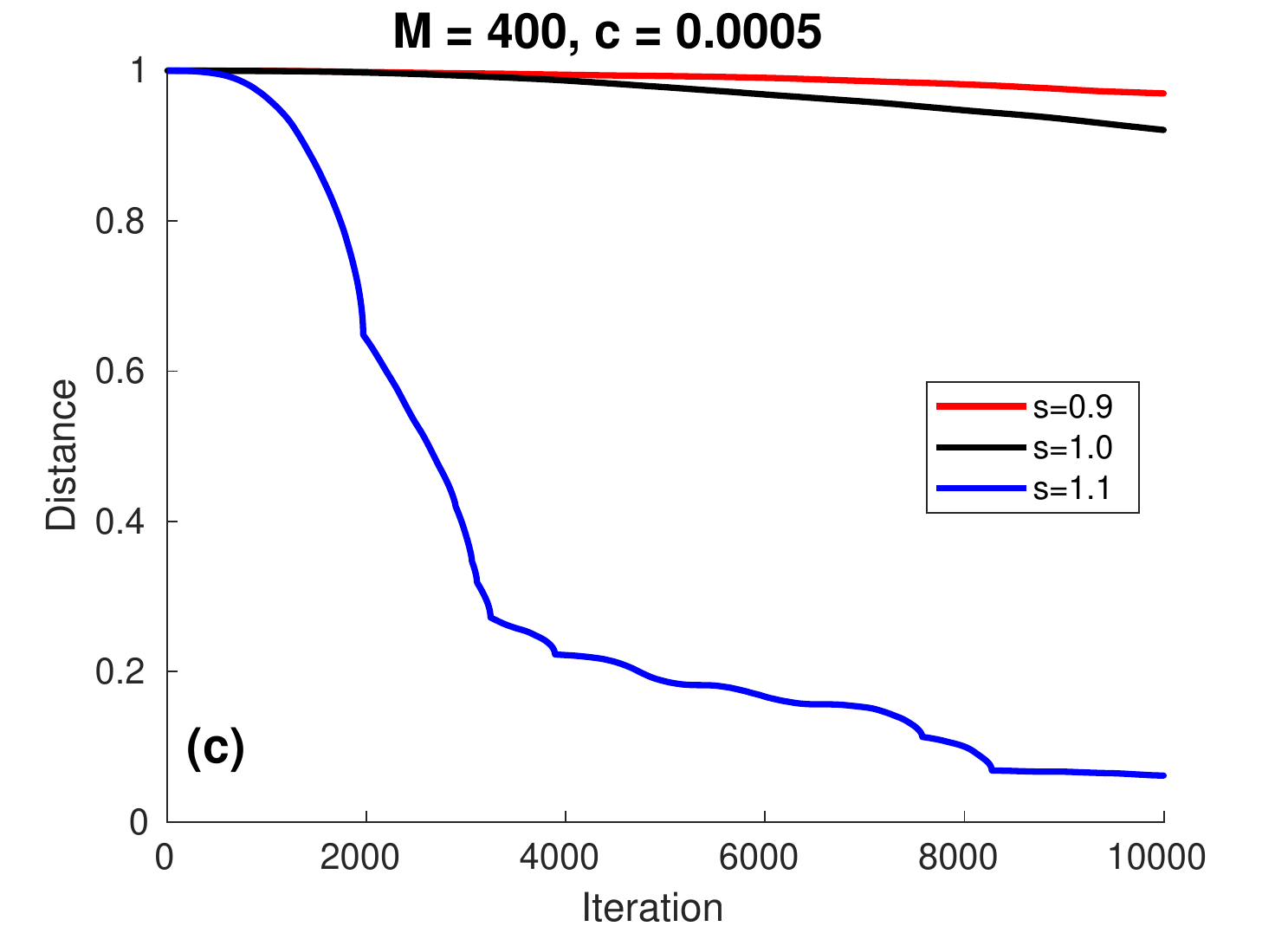}
\includegraphics[width=2.63in,height=1.68in]{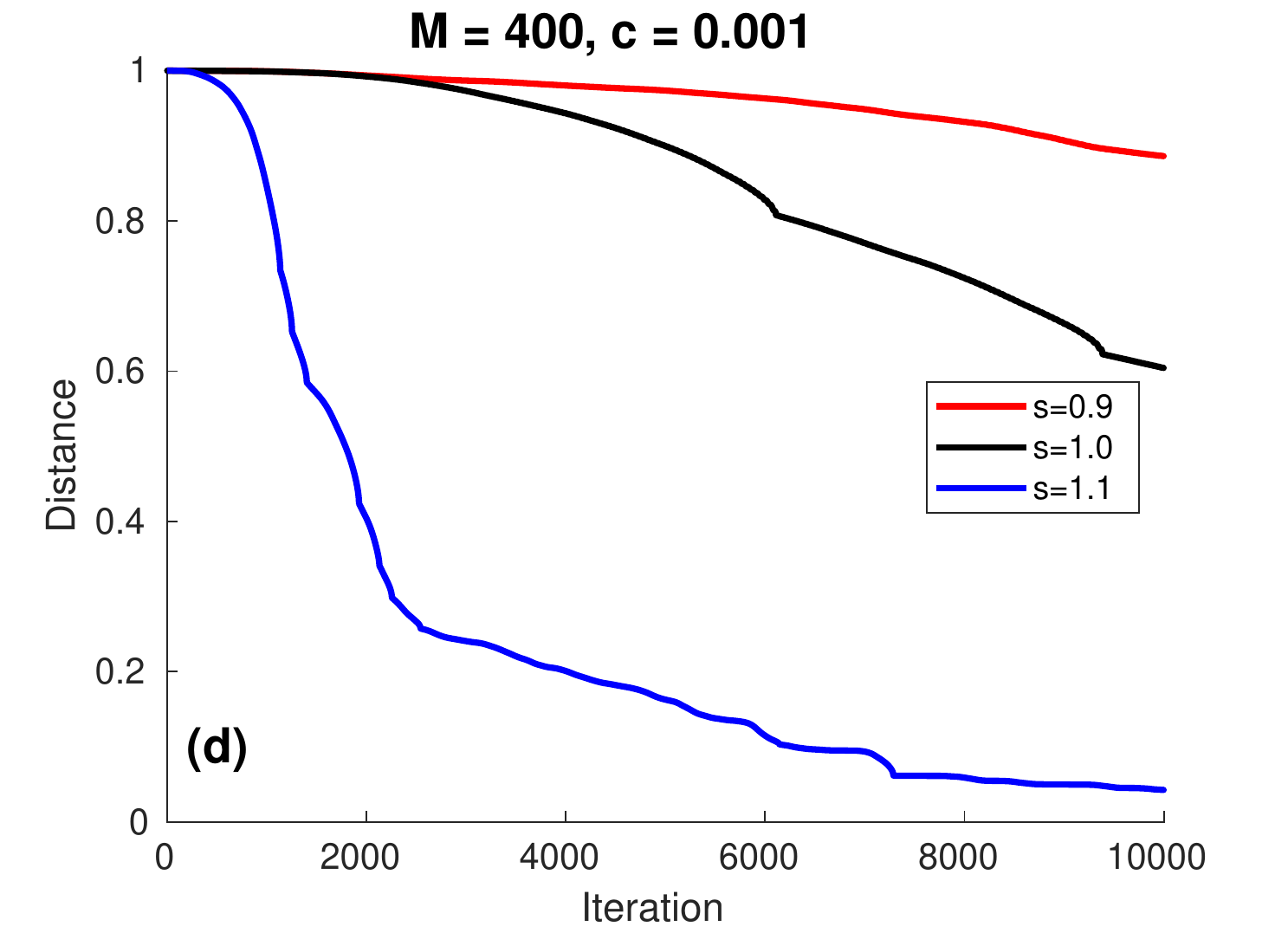}
\caption{A comparison of the effects that different choices of the vector $v$ have on the computation of the distance values $D_{s,\epsilon}^n.$ [LEFT] We consider the plots for (a) $M=300$ and $c=0.0005,$ (b) $M=300$ and $c=0.001,$ (c) $M=400$ and $c=0.0005,$ and (d) $M=400$ and $c=0.001.$ While there are slight difference in the results, the plots are {\em qualitatively} the same.}
\end{figure}

\subsection{Orthogonality Check}



It is well-known that the forward Gram-Schmidt procedure loses orthogonality in many calculations, which can cast doubt on our proposed distance calculations. In order to demonstrate the accuracy of our algorithm, we stored the entire Krylov subspace generated on a smaller problem instance ($n = 150$) and stored these as columns of a matrix $K.$ Then, we computed the value $Q = \|K^TK - I\|_\infty,$ which deviates from zero in proportion to the loss of orthogonality. The values of $Q$ resulting from various $c-$values employed in our algorithm are presented in Table 1, for two different choices of the truncation parameter, $M,$ and $s-$values representing superdiffusion, classical diffusion, and subdiffusion.

Two rows of results are shown for each $s-$value. The first row demonstrates the $Q-$values for the parameter regimes employed in most of our computational experiments, while the second demonstrates the orthogonality preservation for larger disorders. The results in both rows show that the forward Gram-Schmidt vectors are nearly orthogonal. The first row for each $s-$value shows very little change within the $c-$values considered, while the orthogonality seems to decrease as $c$ grows. 



\begin{table}[H]
\centering
\begin{tabular}{|c|||c||c|llllll|}
\hline\hline
$M$ & $s$ & \multicolumn{6}{c}{} & \\
\hline
\multirow{12}{*}{$100$}
& \multirow{4}{*}{$0.9$} 
& c & 0.0 & 0.0001 & 0.0005 & 0.001 & 0.005 & 0.01 \\
& & Q & 5.94e-13 & 7.26e-13 & 1.09e-12 & 1.27e-12 & 1.49e-12 & 1.16e-12 \\
& & c & 0.5 & 1.0 & 1.5 & 2.0 & 2.5 & 3.0 \\
& & Q & 7.27e-13 & 9.54e-13 & 2.28e-11 & 2.89e-10 & 1.04e-10 & 2.37e-08 \\
\cline{2-9}
& \multirow{4}{*}{$1$}
& c & 0.0 & 0.0001 & 0.0005 & 0.001 & 0.005 & 0.01 \\
& & Q & 4.53e-14 & 8.06e-13 & 7.14e-13 & 7.96e-13 & 5.75e-13 & 6.33e-13 \\
& & c & 0.5 & 1.0 & 1.5 & 2.0 & 2.5 & 3.0 \\
& & Q & 1.57e-12 & 4.22e-11 & 7.78e-09 & 5.69e-09 & 1.61e-09 & 5.27e-08 \\
\cline{2-9}
& \multirow{4}{*}{$1.1$}
& c & 0.0 & 0.0001 & 0.0005 & 0.001 & 0.005 & 0.01 \\
& & Q & 7.26e-13 & 9.09e-13 & 2.03e-12 & 1.93e-12 & 7.77e-13 & 7.48e-13 \\
& & c & 0.5 & 1.0 & 1.5 & 2.0 & 2.5 & 3.0 \\
& & Q & 1.44e-13 & 8.29e-12 & 4.18e-10 & 1.72e-10 & 5.16e-08 & 1.50e-08 \\
\hline\hline
\multirow{12}{*}{$500$}
& \multirow{4}{*}{$0.9$} 
& c & 0.0 & 0.0001 & 0.0005 & 0.001 & 0.005 & 0.01 \\
& & Q & 5.94e-13 & 6.39e-13 & 6.73e-13 & 3.66e-13 & 8.65e-13 & 7.59e-13 \\
& & c & 0.5 & 1.0 & 1.5 & 2.0 & 2.5 & 3.0 \\
& & Q & 3.86e-13 & 4.60e-12 & 4.06e-12 & 3.73e-12 & 5.31e-11 & 1.21e-10 \\
\cline{2-9}
& \multirow{4}{*}{$1$}
& c & 0.0 & 0.0001 & 0.0005 & 0.001 & 0.005 & 0.01 \\
& & Q & 4.53e-14 & 8.06e-13 & 7.14e-13 & 7.96e-13 & 5.75e-13 & 6.33e-13 \\
& & c & 0.5 & 1.0 & 1.5 & 2.0 & 2.5 & 3.0 \\
& & Q & 1.57e-12 & 4.22e-11 & 7.78e-09 & 5.69e-09 & 1.61e-09 & 5.27e-08 \\
\cline{2-9}
& \multirow{4}{*}{$1.1$}
& c & 0.0 & 0.0001 & 0.0005 & 0.001 & 0.005 & 0.01\\
& & Q & 1.85e-12 & 5.28e-12 & 5.34e-12 & 4.07e-12 & 3.19e-12 & 7.90e-12 \\
& & c & 0.5 & 1.0 & 1.5 & 2.0 & 2.5 & 3.0 \\
& & Q & 3.61e-12 & 1.86e-10 & 2.24e-10 & 3.76e-09 & 9.38e-10 & 7.67e-08 \\
\hline\hline
\end{tabular}

\caption{Presentation of the loss of orthogonality due to the forward Gram-Schmidt procedure via the $\infty-$norm of the matrix of orthogonalized vectors, when $n=150.$ We present multiple values of each of the parameters employed in the computational algorithm.}
\end{table}

\section{Conclusions and Future Work}



In this work, we presented novel results regarding the properties of the discrete fractional Laplacian, $(-\Delta)^s,\ s\in(0,2),$ and explored the transport behavior of a newly defined fractional Sch{\"o}dinger operator using a spectral technique. We introduced known results for $(-\Delta)^s,\ s\in(0,1],$ and developed the analogous results for the subdiffusive operator, $s\in(1,2).$ In particular, we derived an explicit representation for the action of the subdiffusive discrete fractional Laplacian and provided convergence results with respect to the parameter $s.$ Physical interpretations of this nonlocal operator were provided via probabilistic methods, and by relating the parameter $s$ to the asymptotic mean squared displacement in the anomalous diffusion regime. These initial discussions provide the theoretical framework to justify the applicability of the proposed operator to both quantum mechanical and classical transport problems in the physical world. 

The primary contribution of this work is proposing a numerical approach combining the existing spectral method (introduced in \cite{Liaw2013}) with our newly defined discrete fractional Schrödinger operator given in \R{schro1}. Preliminary results from our simulations demonstrate a clear qualitative difference in the transport behavior of a one-dimensional disordered system as the fraction  is varied. In particular, for $s>1$ ($s<1$) the corresponding transport behavior is more (less) localized than the classical case given by $s=1.$ 
For the classical case, $s=1,$ it is well known that all transport under \R{schro1} is localized for any nonzero disorder. Using this result as a baseline, our presented results suggest that non-localized transport can exist for systems modeled by superdiffusion. The numerical and physical conditions needed for such a phenomenon to occur will be characterized in our next article. The forthcoming work will provide a detailed study on the transport behavior of the one-dimensional system as a function of the various simulation and operator parameters, which determine the disorder concentration and nature of the nonlocal interactions. The observed results will then be used to identify the applicability of the proposed numerical approach to the study of transport phenomena in various physical systems. Limitations of the calculation will also be discussed. 

Our other future work will pursue several different paths, with the ultimate goal of further developing a powerful computational tool for the study of transport behavior in complex disordered media. To extend the findings of the present work, we will provide a detailed discussion on physical interpretation and proper scaling of the proposed technique. Such considerations will allow for the investigation of open questions regarding the parameters needed for de-localized energy states to exist in complex media, such as strongly coupled systems, multi-component flows, and plasmas. In these systems, the existence of extended states in the anomalous diffusion regime will be discussed as a mechanism for the onset of various dynamical phenomena, including turbulence, streaming instabilities, and global vorticity. Finally, we plan to extend this work to higher-dimensional settings. This direction requires both theoretical and computational endeavors that can greatly impact the scientific community.

\section*{Acknowledgments}

\noindent This work was supported by the NSF-DMS grant number 1802682 (CDL), NASA grant number 1571701 (TWH and LSM), NSF grant numbers 1707215 (LSM and TWH), 1740203 (TWH and LSM), and 1903450 (EGK, JLP, CDL, and LSM).

\section*{References}


\bibliographystyle{unsrt}
\bibliography{Anderson_Bib}

\end{document}